\documentclass[11pt,oneside]{article}
\usepackage{graphicx,hyperref,color}
\usepackage[utf8]{inputenc}
\usepackage{marvosym}
\usepackage{cite}
\usepackage{verbatim}
\usepackage{amsmath,amssymb,amsthm}
\usepackage{amsmath,amssymb}
\usepackage[stable]{footmisc}
\usepackage{float}
\usepackage{url}
\usepackage{xspace}
\usepackage{paralist}
\usepackage{multienum}
\usepackage{enumitem}
\usepackage{thmtools}
\usepackage{thm-restate}
\usepackage{hyperref}
\usepackage{cleveref}
\usepackage[normalem]{ulem}
\usepackage{authblk}
\usepackage{algorithm}
\usepackage[noend]{algpseudocode}%
\usepackage{fullpage}
\definecolor{orange}{rgb}{1,0.5,0}

\newcommand{\avgH}{|V(H)|\sqrt{\log |V(H)|}}
\newcommand{\poly}{\mathrm{poly}}
\newcommand{\opt}{\textsc{O}}
\newcommand{\loc}{\textsc{L}}

\newcommand{\ex}{\textsc{Ex}}
\newcommand{\cdiv}{c_{\mathrm{div}}}
\newcommand{\ce}{c_{\mathrm{ex}}}

\newtheorem{theorem}{Theorem}[section]

\newtheorem{lemma}[theorem]{Lemma}
\newtheorem{claim}[theorem]{Claim}

\newtheorem{definition}[theorem]{Definition}
\newtheorem{observation}[theorem]{Observation}

\title{
Local Search is a PTAS for Feedback Vertex Set in \\ Minor-free Graphs\footnote{This material is based upon work supported by the National Science Foundation under Grant No.\ CCF-1252833.}
}
\author[1]{Hung Le}
\author[2]{Baigong Zheng}
\affil[1,2]{Oregon State University}
\affil[1]{\texttt{lehu@onid.oregonstate.edu}}
\affil[2]{\texttt{zhengb@oregonstate.edu}}
\date{}

\begin{document}
\maketitle

\begin{abstract}
We show that a simple local search gives a PTAS for the Feedback Vertex Set (FVS) problem in minor-free graphs. An efficient PTAS in minor-free graphs was known for this problem by Fomin, Lokshtanov, Raman and Sauraubh~\cite{FLRS11}. However, their algorithm is a combination of many advanced algorithmic tools  such as contraction decomposition framework introduced by Demaine and Hajiaghayi~\cite{DH05}, Courcelle's theorem~\cite{Courcelle90} and the Robertson and Seymour decomposition~\cite{RS03}. In stark contrast, our local search algorithm is very simple and easy to implement. It keeps exchanging a constant number of vertices to improve the current solution until a local optimum is reached.  Our main contribution is to show that the local optimum only differs the global optimum by $(1+\epsilon)$ factor.
\end{abstract}

\section{Introduction}

Given an undirected graph, the {\em Feedback Vertex Set} (FVS) problem asks for a minimum set of vertices whose removal makes the graph acyclic. This problem arises in a variety of applications, including deadlock resolution, circuit testing, artificial intelligence, and analysis of manufacturing processes~\cite{ENSZ00}. Due to its importance, the  problem has been studied for a long time. It is one of Karp's 21 NP-complete problems~\cite{Karp72} and is still NP-hard even in planar graphs~\cite{Yannakakis78}. It is one of the two problems that motivates the development of the seminal contraction decomposition framework for designing polynomial time approximation schemes\footnote{ A polynomial-time approximation scheme for a minimization problem is an algorithm that, given a fixed constant $\epsilon > 0$, runs in polynomial time and returns a solution within $1 +\epsilon$ of optimal. } (PTASes) for many optimization problems in planar graphs~\cite{DH05}.

In general graphs, the current best approximation ratio for the FVS problem is 2 due to Becker and Geiger~\cite{BG96} and Bafna, Berman and Fujito~\cite{BBF99}. For some special classes of graphs, better approximation algorithms are known. Kleinberg and Kumar~\cite{KK01} gave the first PTAS for the FVS problem in planar graphs, followed by an efficient PTAS\footnote{A PTAS is efficient if the running time is of the form $2^{\poly(1/\epsilon)}n^{O(1)}$.} by Demaine and Hajiaghayi~\cite{DH05} which is generalizable to bounded genus graphs and single-crossing-minor-free graphs.  Recently, Cohen-Addad et al.~\cite{CCKMM16} gave a PTAS for the weighted version of this problem in bounded-genus graphs. By generalizing the contraction decomposition of Demaine and Hajiaghayi  to minor-free graphs,  Fomin, Lokshtanov, Raman and Sauraubh~\cite{FLRS11} obtained a PTAS for the FVS problem in this class of graphs. A graph is \emph{$H$-minor-free}, or simply minor-free, if it excludes some fixed graph $H$ as a minor. In this work, we assume that $|V(H)|$ is a constant.  We note that the class of minor-free graphs are vastly bigger than planar graphs and bounded-genus graphs. A typical example is the complete bipartite graph $K_{3,n}$ which has unbounded genus but is $K_5$-minor-free. In Section~\ref{sec:negative}, we show that in some sense, minor-free graphs are the limit for which we are still able to obtain a PTAS for this problem. 

A common theme in all known algorithms is complication in both implementation and analysis. The algorithm of Kleinberg and Kumar~\cite{KK01} is obtained by recursively applying the planar separator theorem by Lipton and Tarjan~\cite{LT79} and analyzing several special cases. The algorithm by Demaine and Hajiaghayi~\cite{DH05} employs the primal-dual relationship of planar graphs to decompose the graphs into several bounded treewidth instances, then applies dynamic programming (DP) to solve the FVS problem on bounded treewidth graphs. DP on bounded treewidth graphs is a very strong algorithmic tool. However, the implementation details typically are quite complicated. Additionally, the NP-hardness complexity of finding a tree decomposition of minimum width in planar graphs is still a long standing open problem. The algorithm of Cohen-Addad et al.~\cite{CCKMM16} for bounded-genus graphs is not simpler and has worst running time; however, it can work with node-weighted graphs. Given the complicated nature of the algorithms for planar and bounded-genus graphs, it is not surprising that the technical level of the algorithm by  Fomin, Lokshtanov, Raman and Sauraubh~\cite{FLRS11} for minor-free graphs is much higher. It uses  advanced tools such as  Courcelle's theorem~\cite{Courcelle90} and the Robertson and Seymour decomposition~\cite{RS03}. We note that the decomposition of Robertson and Seymour was built through a series of papers which span 20 years with several hundred pages~\cite{RS83,RS03}. Thus, even understanding Robertson and Seymour decomposition is a real challenge, let alone implementing it. All of this motivates our current work.

We show that a simple local search algorithm gives a PTAS for the FVS problem in minor-free graphs. The algorithm is depicted in Algorithm~\ref{alg:localsearch}. Intuitively, the local search algorithm starts with an arbitrary solution for the problem and tries to change a constant number (depending on $\epsilon$) of vertices in the current solution to obtain a better solution. The algorithm outputs the current solution when it cannot obtain a better solution in this way.

Local search is among the most successful heuristics in combinatorial optimization, partly due to its simplicity. It has been applied to  scheduling, graph coloring, graph partitioning,  Hopfield networks; we refer readers to the monograph by Michiels, Aarts and Korst~\cite{MAK10} for more details. However, one of the hardest questions regarding local search is the performance guarantee. We provide an answer this question for the FVS problem. The analysis of our algorithm is simple, but non-trivial:  it only uses two  well-known properties of $H$-minor-free graphs as black boxes, namely sparsity and separability,  and can be described in about four pages.  A key ingredient in our analysis is the introduction of Steiner vertices into the construction of exchange graphs which is different from all previous works~\cite{CH09,MR10,CG15}; we defer further details of this discussion to Subsection~\ref{subsec:our-analysis}.

\begin{algorithm}
\caption{\textsc{LocalSearch}$(G, \epsilon)$}
\label{alg:localsearch}
\begin{algorithmic}[1]
\State $S\leftarrow$ an arbitrary solution of $G$
\State $c \leftarrow$ a constant depending on $\epsilon$
\While{
there is a solution $S'$ such that $| S \setminus S'| \le c$, $|S' \setminus S| \le c$ and $|S'| < |S|$
}
\State $S \leftarrow S'$
\EndWhile
\State output $S$
\end{algorithmic}
\end{algorithm}

\begin{theorem}\label{thm:main} For any fixed $\epsilon > 0$, there is a local search algorithm that finds an $(1+\epsilon)$-approximate solution for the FVS problem in $H$-minor-free graphs with running time $O(n^c)$ where $c = \frac{\poly(|V(H)|)}{\epsilon^2} $.
\end{theorem}

Beside simplicity, our algorithm has two other interesting properties. First, to run the algorithm, we do not need to know beforehand whether the graph under consideration is minor-free or not; it will give a PTAS in case the graph is minor-free. All known algorithms discussed above need to test topological properties of the graph, such as planarity, genus-boundedness or minor-freeness, to be able to decide whether the algorithms are applicable. Except for planarity, other testings are quite expensive~\cite{johnson87,Mohar99}. Second, the dependency of the exponent of  the running time in our algorithm on the size of the minor is $\poly(|V(H)|)$, or $O(|V(H)|^{3/2})$ precisely while the constant behind the big-O in the running time of the algorithm by  Fomin, Lokshtanov, Raman and Sauraubh~\cite{FLRS11} is a tower function of $|V(H)|$.  Even when $|V(H)| = 5$, the constant is still bigger than the size of the universe~\cite{johnson87}.

Perhaps the only drawback of our result is the running time dependency on $\epsilon$, which is roughly $n^{O(\frac{1}{\epsilon^2})}$. However, our result should be seen as the first step toward theoretically understanding of the power of local search for the FVS problem: as long as we are willing to pay for computational time,  we are guaranteed to get better approximation ratio.  For APX-hard problems, such as the FVS problem, there is a limit to which, if one increases the neighborhood size, the gain in approximation is zero or negligible. Thus, a natural question is:  when the input has some structural properties, would it be possible to obtain better approximation ratio when the neighborhood size increases? A yes answer to this question would be quite significant in practice because real instances typically have some structural properties and the local search algorithm does not need to test such properties.  Our Theorem~\ref{thm:main} provides a yes answer to this question, when the structure of the input is minor-free. Also,  in practice, one often runs local search with $c = 4$ or $5$ ($c$ is in line 2 of Algorithm~\ref{alg:localsearch}.). It will be interesting to know, even in planar graphs, when $c = 4$ or $5$, what is the approximation guarantee we can obtain? Indeed, there have been some recent work~\cite{AMM17,BGMR17} toward this direction for optimization problems admitting local search PTASes (with the same running time as our algorithm in Theorem~\ref{thm:main}). Our Theorem~\ref{thm:main} says that there has to be a constant $c$ such that when we apply local search to planar graphs with $c$, we will beat the best known 2-approximation algorithm for general graphs~\cite{BG96,BBF99}. We leave the problem of determining the exact constant $c$ as an open problem for future research. Finally, we would like to point out that local search was experimentally applied to the FVS problem with good results~\cite{ZYZS13,QZ14}. In a certain sense, our result helps justifying for them.

To complement our positive result, we provide several negative results. The work of Har-Peled and Quanrud~\cite{HQ15} shows that local search provides PTASes for several problems, including vertex cover, independent set, dominating set and connected dominating set, in graphs with polynomial expansion (all of these problems are known to have PTASes in minor-free graphs.). Minor-free graphs are a special case of graphs with polynomial expansion. Thus, their work gives a hope that local search can be used to generalize known PTASes for optimization problems from minor-free graphs to graphs of polynomial expansion. However, our first negative result refuses this hypothesis. By a simple reduction, we show that the FVS problem is APX-hard in 1-planar graphs.  Note that $1$-planar graphs also are a special case of graphs of polynomial expansion.  Second, we show that two closely related variants of the FVS problem, namely:  {\em odd cycle transversal} and {\em subset feedback vertex set}, do not have such simple local search PTASes, even in planar graphs. We remark that these two problems are not known to have PTASes   in planar graphs. 

\subsection{Our analysis technique}\label{subsec:our-analysis}

To better put our technique into context, we briefly discuss previous work. Chan and Har-Peled~\cite{CH09} and Mustafa and Ray~\cite{MR10} independently showed that a simple local search gives PTASes for many geometric problems. Cabello and Gajser~\cite{CG15} observed that  the same local search can be used to design PTASes for the maximum independent set, the minimum vertex cover and minimum dominating set problems in minor-free graphs.  Cohen-Addad, Klein and Mathieu~\cite{CKM16} showed that local search yields PTASes for $k$-means, $k$-median and uniform uncapacitated facility location in minor-free graphs. In analyzing local search algorithms, one typically relies on an {\em exchange graph} constructed from the optimal solution\footnote{For $k$-means and $k$-median, the exchange graph is constructed from $\loc$ and a nearly optimal solution $\opt'$, which is obtained by removing some vertices of $\opt$.} $\opt$ and the local search solution $\loc$.  For independent set and vertex cover, the exchange graph is the subgraph induced by $\opt \cup \loc$, and for other problems, the exchange graph is obtained  by contracting each vertex of $V(G)\setminus (O\cup L)$ to a nearest vertex in $\opt \cup \loc$.  Then local properties of these problems naturally appear in the exchange graphs: if we consider a small neighborhood $R$ in the exchange graph and replace the vertices of $\loc$ in $R$ with the vertices of $\opt$ in $R$ and its the boundary, the resulting vertex set is still a feasible solution. By decomposing the exchange graph into small neighborhoods, we can bound the size of $\loc$ by the size of $\opt$ and the total size of the boundaries of these neighborhoods.

However, the FVS problem does not have such local properties and hence, just simply deleting vertices and contracting edges do not give us an exchange graph. This is because for a cycle $C$ in the original graph, the vertex of $\loc$ that covers $C$ may be inside of some neighborhood but the vertex of $\opt$ that covers $C$ may be outside of that neighborhood. One may try to argue the boundary of the neighborhood could cover $C$. But unfortunately, the boundary may not be helpful since the crossing vertices of $C$ and the boundary may not be in both solutions and then they may be deleted or contracted to other vertices.

To solve this problem, we construct an exchange graph with the following property: for any cycle $C$ of the original graph, in our exchange graph, there is (i) a vertex in $\opt \cap \loc \cap C$, or  (ii) an edge between a vertex in $\opt \cap C$ and a vertex in $\loc \cap C$, or  (iii) another cycle $C'$ such that vertices in $C'$ is a subset of vertices in $C$ and $C' \cap (\opt \cup \loc) = C \cap (\opt \cup \loc)$.  Property (i) and/or (ii) are typically achieved in previous analyses~\cite{CG15,CKM16} by vertex deletion or edge contraction.  It is property (iii) that is specific to our problem and is a main challenge.  To additionally achieve this property, we need to introduce vertices, called \emph{Steiner vertices}, that are not in both solutions, into the exchange graph. Meanwhile, we need to guarantee that the number of such vertices is linear to the size of $\opt \cup \loc$.  The linear size bound is essential to the correctness of our algorithm and we prove this size bound by a structural lemma (Lemma~\ref{lm:structure}) which may be of independent interest.

In summary, this is the first time Steiner vertices are proved useful in analyzing local search due to the non-local nature of the FVS problem. Given that many optimization problems, such as  minor covering and packing problems~\cite{FLRS11}, exhibit the same non-local properties, we believe that our technique is useful in studying the local search algorithm for these problems as well.

\section{Preliminaries}

For a graph $G$, we denote the vertex set and the edge set of $G$ by $V(G)$ and $E(G)$, respectively. For a subgraph $H$ of $G$, the {\em boundary} of $H$ is the set of vertices that are in $H$ but have at least one incident edge that is not in $H$. We denote by $int(H)$ the set of vertices of $H$ that are not in the boundary of $H$. The {\em degree} of a vertex is the number of its incident edges. 

A graph $H$ is a \emph{minor} of $G$ if $H$ can be obtained from $G$ by a sequence of vertex deletions, edge deletions and edge contractions. $G$ is \emph{$H$-minor-free}, if $G$ does not contain a fixed graph $H$ as a minor.  We sometimes call  $H$-minor-free graphs \emph{minor-free graphs} when the size of $H$ is not relevant. It is well known~\cite{Kostochka82,Kostochka84} that $H$-minor-free graph is sparse;  an $H$-minor-free graph with $n$ vertices has at most $O(\sigma_H n)$ edges where $\sigma_H = \avgH$. 

A {\em balanced separator} of a graph is a set of vertices whose removal partitions the graph roughly in half.
A separator theorem typically provides bounds for the size of each part and the size of the balanced separator. Usually, the size of the balanced separator is sublinear w.r.t. the size of the graph. Separator theorems have been found for planar graphs~\cite{LT79}, bounded-genus graphs~\cite{GHT84}, and minor-free graphs~\cite{AST90}.
 
An {\em $r$-division} is a decomposition of a graph, which was first introduced by Frederickson~\cite{Frederickson87} for planar graphs to speed up planar shortest path algorithms.
\begin{definition}\label{def:r-div}
For an integer $r$, an $r$-division of a graph $G$ is a collection of edge-disjoint subgraphs of $G$, called {\em regions}, with the following properties:
\begin{itemize}[noitemsep]
\item[1.] Each region contains at most $r$ vertices and each vertex is contained in at least one region.
\item[2.] The number of regions is at most $\cdiv\frac{n}{r}$. 
\item[3.] The number of boundary vertices, summed over all regions, is at most $\cdiv\frac{n}{\sqrt{r}}$.
\end{itemize}
where $\cdiv$ is a constant. 
\end{definition}

We say a graph is {\em $r$-divisible} if it has an $r$-division. A graph is \emph{divisible} if it is $r$-divisible for every $r$. Given any $r$ and a planar graph $G$, Frederickson~\cite{Frederickson87} gave a construction for the $r$-division of $G$ that only relies on the planar separator theorem~\cite{LT79}. It is straightforward to extend the construction to any family of graphs with balanced separators of sublinear size. Since $H$-minor-free graphs are known to have balanced separators~\cite{AST90}, $H$-minor-free graphs are divisible with $\cdiv = \poly(|V(H)|)$.

\section{Exchange graphs imply PTASes by local search}

In this section, we show that if for a minor-free graph $G$, we can construct another graph, called {\em exchange graph}, such that it is divisible, then Algorithm~\ref{alg:localsearch} is a PTAS for the FVS problem. Let $\opt$ be an optimal solution of the FVS problem and $\loc$ be the output of the local search algorithm. We say a vertex $u$ a \emph{solution vertex} if $ u \in \opt \cup \loc$ and a \emph{Steiner vertex} otherwise. Unlike prior works~\cite{CG15,HQ15}, we allow \emph{Steiner vertices} in our exchange graphs.

\begin{definition}\label{def:exch} 
A graph $\ex$ is an exchange graph for optimal solution $\opt$ and local solution $\loc$ of the FVS problem in a graph $G$ if it satisfies the following properties:
\begin{enumerate}
\item[(1)] $\loc \cup \opt \subseteq V(\ex) \subseteq V(G)$.
\item[(2)] $|V(\ex)| \leq \ce (|\loc| + |\opt|)$ for some constant $\ce$. 
\item[(3)] For every cycle $C$ of $G$, there is (3a) a vertex of $C$ in $\opt \cap \loc$ or (3b) an edge $uv \in E(\ex)$ between a vertex $u \in \loc \cap C$ and a vertex $v \in \opt \cap C$ or (3c) a cycle $C'$ of $\ex$ such that $V(C')\subseteq V(C)$ and $C \cap (\opt \cup \loc) = C' \cap (\opt \cup \loc)$.
\end{enumerate}
\end{definition}

We now prove Theorem~\ref{thm:main} given that we can construct a \emph{divisible exchange graph} for $G$. The details of the construction will be given in Section~\ref{sec:ex-const}.

\begin{proof}[Proof of Theorem~\ref{thm:main}]
We set the constant $c$ in line 2 of Algorithm~\ref{alg:localsearch} to be $1/\delta^2$ where $\delta = \frac{\epsilon}{2\cdiv \ce(2+\epsilon)} = O(\frac{\epsilon}{\cdiv \ce})$. Note that $\cdiv$ and $\ce$ are constants in Definition~\ref{def:r-div} and Definition~\ref{def:exch}, respectively. Since in each iteration, the size of the solution is reduced by at least one, there are at most $n$ iterations. Since each iteration can be implemented in $n^{O(c)}$ time by enumerating all possibilities, the total running time is $n^{O(c)} = n^{O(1/\epsilon^2)}$. We now show that the output $L$ has size at most $(1+\epsilon)|\opt|$. 
 
Let $\ex$ be a divisible exchange graph  for  $\opt$ and $\loc$.  We find an $r$-division of $\ex$ for $r = c = 1/\delta^2$. Let $B$ be the multi-set containing all the boundary vertices in the $r$-division. By the third property in Defintion~\ref{def:r-div}, $|B|$ is at most $ \cdiv \frac{|V(\ex)|}{\sqrt{r}}$. By the second property in Definition~\ref{def:exch}, $|V(\ex)| \leq \ce (|\opt| + |\loc|)$. Thus, $|B| \le \cdiv \ce \delta (|\opt| + |\loc|)$. In the following, we will show that:
\begin{align}
|\loc| \le |\opt| + 2 |B|
\label{equ:2}
\end{align}
If so, we have:
\begin{align*}
|L| \leq |\opt| + 2 \cdiv \ce \delta (|\opt| + |\loc|)  = |\opt|  + \frac{\epsilon}{2+\epsilon}(|\opt| + |\loc|)
\end{align*}
that implies $|L| \leq (1+\epsilon)|\opt|$.

To prove Equation~\eqref{equ:2}, we study some properties of $\ex$. For any region $R_i$ of the $r$-division, let $B_i$ be the boundary of $R_i$ and $M_i = (\loc \setminus R_i) \cup (\opt \cap R_i) \cup B_i$.

\begin{claim}
$M_i$ is a feedback vertex set of $G$.
\end{claim}  
\begin{proof}
For a contradiction, assume that there is a cycle $C$ of $G$ that is not covered by $M_i$. Then $C$ does not contain any vertex of $\loc \setminus R_i$, $\opt \cap R_i$ and $B_i$.
So $C$ can only be covered by some vertices of $(\loc \setminus \opt) \cap int(R_i)$ and some vertices of $\opt \setminus (\loc \cup R_i)$. This implies that $C$ does not contain any vertex of $\opt \cap \loc$ and there is no edge in $\ex$ between $C\cap \opt$ and $C \cap \loc$.
By the third property of exchange graph, there must be a cycle $C'$ in $\ex$ such that $V(C') \subseteq V(C)$ and $C \cap (\opt \cup \loc) = C' \cap (\opt \cup \loc)$.
Let $u$ be the vertex of $(\loc \setminus \opt) \cap int(R_i)$ in $C$ and $v$ be the vertex of $\opt \setminus (\loc \cup R_i)$ in $C$.
Then cycle $C'$ contains both $u$ and $v$, which implies $C'$ crosses the boundary of $R_i$, that is $C' \cap B_i \ne \emptyset$.
Let $w$ be a vertex in $C' \cap B_i$, then $w$ also belongs to $C$ in $G$.
This implies $M_i$ contains a vertex of $C$, a contradiction.
\end{proof}
By the construction of $M_i$, 
we know the difference between $\loc$ and $M_i$ is bounded by the size of the region $R_i$, that is $r$. Recall that $c = r = 1/\delta^2$. Since $\loc$ is the output of Algorithm~\ref{alg:localsearch}, it cannot be improved by changing at most $r$ vertices. Thus, we have $|\loc| \le |M_i|$.
By the construction of $M_i$, this implies 
$$|\loc \cap R_i| \le |M_i \cap R_i| \le |\opt \cap int(R_i)| + |B_i|.$$
Thus, we have:
$$|\loc \cap int(R_i)| \le |\loc \cap R_i| \le |\opt \cap int(R_i)| + |B_i|.$$  
Since $int(R_i)$ and $int(R_j)$ are vertex-disjoint for any two distinct $i$ and $j$, by summing over all regions in the $r$-division, we get
$$|\loc| - |B| \le \sum_i |\loc \cap int(R_i)| \le \sum_i (|\opt \cap int(R_i)| + |B_i|) \le |\opt| +  |B|.$$
This proves Equation~\eqref{equ:2}. \qedhere

\end{proof}
\section{Exchange graph construction}\label{sec:ex-const}

Recall that $\sigma_H  = \avgH$ is the sparsity of $H$-minor-free graphs.  In this section, we will show that $H$-minor-free graphs have divisible exchange graphs for the FVS problem with $\ce = O(\sigma_H)$. We construct the exchange graph in three steps:

\begin{description}
\item[Step 1] We  delete all edges in $G$ that are incident to vertices of $\opt \cap \loc$. We then remove all components of $G$ that do not contain any solution vertex. Note that the removed components are acyclic.
\item[Step 2] Let $v \in V(G)\setminus (O\cup L)$ be a non-solution vertex of degree at most $2$. Recall that isolated vertices are removed in Step 1. If $v$ has degree $1$, we simply remove $v$ from $G$. If $v$ has degree 2, we remove $v$ from $G$ and add an edge between two neighbors of $v$ in $G$. We can view this step in terms of contraction: we contract edges that have an endpoint that is not a solution vertex and has degree at most two until there is no such an edge left. Since $\loc$ and $\opt$ are feedback vertex sets of $G$, every cycle after the contraction must contain a vertex in $\loc$ and a vertex in $\opt$.  Since edges incident to vertices of $\opt \cap \loc$ are removed, there is no self-loop after this step. 
\item[Step 3] We keep the graph simple by removing all but one edge in each maximal set of parallel edges. 
\end{description}

Let $K$ be the resulting graph.  Since $K$ is a minor of $G$, it is $H$-minor-free and thus, divisible. It remains show that $K$ satisfies three properties in Definition~\ref{def:exch}. Property (1) is obvious because we never delete a vertex in $L\cup O$ from $K$.  To show property (3), let $C$ be a cycle of $G$. If any edge of $C$ is removed in Step 1, $C$ must contain a vertex in $\opt \cap \loc$; implying (3a).  Thus, we can assume that no edge of $C$ is deleted after Step 1. Since contraction does not destroy cycles, after the contraction in Step 2, there is a cycle $C'$ such that $V(C')\subseteq V(C)$. If $|V(C')| = 2$ ($C'$ is a cycle of two parallel edges), then (3b) holds. Thus, we can assume that every edge of $C'$ remains  intact after removing parallel edges. But that implies (3c) since we never remove solution vertices from $G$. Thus, $K$ satisfies property (3). 

The most challenging part is showing property (2) in Definition~\ref{def:exch}, that is, $|V(K)| \leq O(\sigma_H)(|\loc| + |\opt|)$. By Step 2, we have:

\begin{observation} \label{obs:K-Steiner-deg-3} 
Every Steiner vertex of $K$ has degree at least 3.
\end{observation}

Since $\opt \cup \loc$ is a feedback vertex set of $K$, $K\setminus (\opt \cup \loc)$ is a forest $F$ containing only Steiner vertices. For each tree $T$ in $F$, we define the \emph{degree} of $T$, denoted by $\deg_K(T)$, to be the number of edges in $K$ between $T$ and $\opt \cup \loc$.  

\begin{claim}\label{clm:size-T-vs-deg-T}
$|V(T)| \leq \deg_K(T)$.
\end{claim}
\begin{proof}
Let $T'$ be obtained from $T$ by adding every edge $uv$ to $T$ where $u \in V(T)$ and $v \in O\cup L$. Observe that no vertex in $(O\cup L)\setminus (O\cap L)$ can be adjacent to more than one vertex in $T$ since otherwise, there would be a cycle that contains vertices from $L$ or $O$ only, contradicting that $L$ and $O$ are feedback vertex sets.  Since vertices in $L\cap O$ are isolated in $K$,  $T'$ must be a tree.  Let $\ell(T')$ be the number of leaves of $T'$.  By Step 2, leaves of $T'$ are vertices in $O\cup L$. Thus, $\deg_K(T) = \ell(T')$. Since every internal vertices of $T'$ has degree at least three, $|V(T)| \leq \ell(T')$ which implies the claim.
\end{proof}

We contract each tree $T$ of $F$ into a single Steiner vertex $s_T$. Let $K'$ be the resulting graph. We observe that:
\begin{observation}\label{obs:K'-simple} 
$K'$ is simple. 
\end{observation}
\begin{proof}
Since every cycle of $K$ must contain a vertex from $\loc$ and a vertex from $\opt$, there cannot be any solution vertex in $K$ that is adjacent to more than one vertex of a tree $T$ of $F$. So there cannot be parallel edges in $K'$.
\end{proof}

To bound the size of $K'$, we need the following structural lemma. We remark that this lemma holds for general graphs. 

\begin{lemma}\label{lm:structure}
For a graph $G$ and any two disjoint nonempty vertex subsets $A$ and $B$, let $D = V(G)\setminus (A\cup B)$. If (i) $D$ is an independent set, (ii) every vertex in $D$ has degree at least 3 in $G$ and (iii) every cycle $C$ contains  at least one vertex in $A$ and at least one vertex in $B$, then $|V(G)| \leq 2(|A| + |B|)$.
\end{lemma}
\begin{proof}
 We remove every edge that only has endpoints in $A\cup B$ and let the resulting graph be $G'$. Then $G'$ is a bipartite graph with $A\cup B$ in one side and $D$ in the other side since $D$ is an independent set.  Let $D_A$ ($D_B$) be the subset of $D$ containing every vertex that has at least two neighbors in $A$ ($B$). Since every vertex of $D$ has degree at least $3$, we have $D_A\cup D_B = D$. 

Let $H_A$ be the subgraph of $G'$ induced by $A\cup D_A$. Then $H_A$ is acyclic since otherwise every cycle of $H_A$ would correspond to a cycle in $G$ that does not contain any vertex in $B$. We now construct a graph $H_A^*$ on vertex set $A$. For each vertex $v \in D_A$, we choose any two neighbors $x$ and $y$ of $v$ in $A$ and add an edge between $x$ and $y$ in $H_A^*$. By construction, there is a one-to-one mapping between edges of $H_A^*$ and vertices of $D_A$.

Since $H_A$ is acyclic, $H_A^*$ is also acyclic. Thus, $|E(H_A^*)| \leq V(H_A^*) = |A|$. That implies $|D_A| \leq |A|$. By a similar argument, we can show that $|D_B| \leq |B|$. Thus, $|D| = |D_A\cup D_B| \leq |A| + |B|$ which implies the lemma.
\end{proof}

Let $Z$ be an arbitrary component of $K'$ that contains at least one Steiner vertex. 
Then two sets $V(Z) \cap \opt $ and $V(Z) \cap \loc$ must be disjoint since any vertex in $\opt\cap \loc$ is isolated in $K'$. If any of two sets $V(Z) \cap \opt $  and $V(Z) \cap \loc $, say $V(Z) \cap \opt $, is empty, then $Z$ must be a tree. By Step 2, leaves of $Z$ are in $L$. Thus, $|V(Z)| \leq |V(Z)\cap \loc|$ since internal vertices of $Z$ have degree at least $3$. Otherwise, both  $V(Z) \cap \opt $  and $V(Z) \cap \loc $ are non-empty. Let $X$ be the set of Steiner vertices in $Z$. By the construction of $K'$, $X$ is an independent set of $Z$. By Observation~\ref{obs:K-Steiner-deg-3}, every vertex of $X$ has degree at least 3.  
So we can apply Lemma~\ref{lm:structure} to $X$, $V(Z) \cap \opt$ and $V(Z) \cap \loc$, and obtain $|V(Z)| \le 2(|V(Z) \cap \opt| + |V(Z) \cap \loc|) = 2(|V(Z) \cap \opt| + |V(Z) \cap (\loc \setminus \opt)|)$. Note that this bound holds trivially if $Z$ does not contain any Steiner vertex. In both cases, $|V(Z)| \leq 2(|V(Z) \cap \opt| + |V(Z) \cap (\loc \setminus \opt)|)$. Summing over all components of $K'$, we have $|V(K')| \le 2(|V(K') \cap \opt| + |V(K') \cap (\loc\setminus \opt)|) \le 2 (|\opt| + |\loc|)$.
Since $K'$ is a minor of $G$, it is also $H$-minor-free. Thus, 
$|E(K')| =  O(\sigma_H|V(K')|) = O(\sigma_H)(|\opt| + |\loc|)$. We now ready to bound the size of $V(K)$. We have:
\begin{equation}
\begin{split}
|V(K)\setminus (\opt  \cup \loc)|  &= \sum_{T\in F}|V(T)|  \leq \sum_{T\in F}\deg_K(T)  \quad \mbox{(Claim~\ref{clm:size-T-vs-deg-T})}\\
&= \sum_{T\in F}\deg_{K'}(s_T)   \\
&\leq |E(K')| \quad  \mbox{ ( since $\{s_T| T\in F\}$ is an independent set)}\\
&= O(\sigma_H)(|\opt| + |\loc|) 
\end{split}
\end{equation}
That implies $V(K) \leq  O(\sigma_H)(|\opt| + |\loc|)$. Thus $K$ satisfies property (2) in Definition~\ref{def:exch} with $\ce = O(\sigma_H)$.

\section{Negative results}\label{sec:negative}

In this section, we show some negative results for the FVS problem and two closely related problems: odd cycle tranversal and subset feedback vertex set.  The odd cycle transversal (also called {\em bipartization}) problem asks for a minimum set of vertices in an undirected graph whose removal results in a bipartite graph.  Given an undirected graph and a subset $U$ of vertices, the subset feedback vertex set problem asks for a minimum set $S$ of vertices such that after removing $S$ the resulting graph contains no cycle that passes through any vertex of $U$. 

We first show that the FVS problem is APX-hard in $1$-planar graphs.  A graph is {\em 1-planar} if it can be drawn in the Euclidean plane such that every edge has at most one crossing. 

\begin{theorem}\label{thm:apx-hard-fvs}
Given a graph $G$, we can construct a 1-planar graph $H$ in polynomial time, such that $G$ has a feedback vertex set of size at most $k$ if and only if $H$ has a feedback vertex set of size at most $k$.
\end{theorem}
\begin{proof} 
Consider a drawing of $G$ on the plane where each pair of edges can cross at most once. For each crossed edge $e$ in $G$, we subdivide $e$ into edges so that there is exactly one crossing per new edge. Let $H$ be the resulting graph. By construction, graph $H$ is 1-planar.

Let $n$ be the size of $G$. Since there are at most $O(n^2)$ crossings per edge in the drawing, the size of $H$ is at most $O(n^4)$. Sine we only subdivide edges, there is a one-to-one mapping between cycles of $G$ and cycles of $H$. It is straightforward to see that any feedback vertex set of $G$ is also a feedback vertex set of $H$. 

Let $S$ be a feedback vertex set of $H$. If $S \subseteq V(H)\cap V(G)$, then it is also a feedback vertex set for $G$. Otherwise, let $v \in V(H)\setminus V(G)$ be a vertex in $S$. Then $v$ must be added to subdivide an edge, say $e$, in $G$.  We remove $v$ from $S$ and add an arbitrary endpoint of $e$ in $G$ to $S$. Then $S$ is still a feedback vertex set for $H$. We repeat this process until $S$ is a subset of $V(H)\cap V(G)$. Observe that $S$ is a feedback vertex set of size at most $k$ for $G$. Thus, the lemma holds. 
\end{proof}

Since the FVS problem is APX-hard in general graphs (by an approximation preserving reduction~\cite{Karp72} from vertex cover problem, which is APX-hard~\cite{DS05}), Theorem~\ref{thm:apx-hard-fvs} implies that FVS is APX-hard in 1-planar graphs.

To show that simple local search cannot give a constant approximation for the odd cycle transversal problem and the subset feedback vertex set problem, we construct a counter-example from a  $k \times k$ grid as shown in Figure~\ref{fig: example}.

\begin{figure}
\centering
\includegraphics[scale = 1]{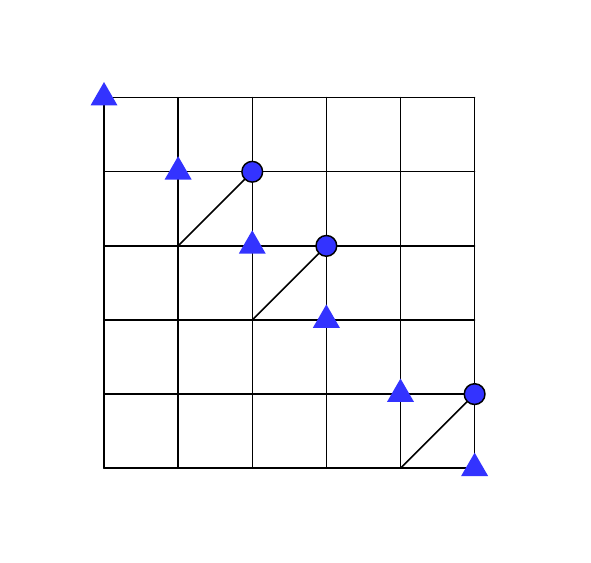}
\includegraphics[scale = 1]{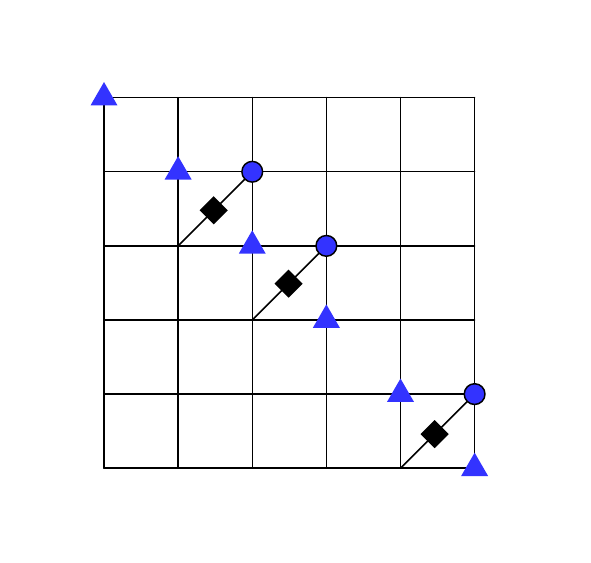}
\caption{Counterexamples for local search on odd cycle transversal and subset feedback vertex set. Circle vertices represent vertices of the optimal solution, and triangle vertices represent vertices of the local search solution. The grid could be arbitrarily large. We add one edge in some diagonal cells of the grid.
Left: a counterexample for the odd cycle transversal problem. Since any grid is bipartite and does not contain any odd cycle, any odd cycle in the example must contain an edge in the diagonal cell. 
All the vertices in the diagonal, represented by triangles, give a solution that is locally optimal, that is, we cannot improve this solution by changing a small number of vertices. This is because each triangle vertex and each new edge, together with some other edges, can form at least one odd cycle in the graph. For a constant $c$ that is smaller than the size of optimal solution, if we remove $c$ triangle vertices, say $V'$, in the locally optimal solution, there will be $c$ vertex-disjoint odd cycles in the resulting graph, each of which contains one removed triangle. Thus, there is no subset of size less than $c$ that can replace $V'$. Then the ratio between the two solutions could be arbitrarily big if the gird is arbitrarily big and the number of added diagonal edges is super-constant and sublinear to the size of the diagonal.
Right: a counterexample for the subset feedback vertex set problem. The diamonds represent vertices in the given set $U$. Similarly, any cycle through a given vertex must contain the two edges in the diagonal cell. By the same reason, the local search solution cannot be improved.
}
\label{fig: example}
\end{figure}

\newpage

\bibliographystyle{plain}
\bibliography{local}
 
\end{document}